\documentclass[11pt]{article}
\usepackage[a4paper]{geometry}
\usepackage{amsfonts, amsmath, amssymb, amsthm, graphicx, caption, authblk, multirow, makecell, framed, float, xcolor, enumitem, tikz, hyperref}
\theoremstyle{definition}

\newtheorem{lemma}{Lemma}

\theoremstyle{remark}

\newcommand*{\mybox}[1]{%
  \framebox{\raisebox{0cm}[0.5\baselineskip][0.05\baselineskip]{%
    \hbox to 0.1cm {\hss#1\hss}}}\hspace{0.05cm}}

\begin{document}
\title{Physical Zero-Knowledge Proof for Ripple Effect\thanks{A preliminary version of this paper \cite{ripple} has appeared at WALCOM 2021.}}
\author[1]{Suthee Ruangwises\thanks{\texttt{ruangwises@gmail.com}}}
\author[1]{Toshiya Itoh\thanks{\texttt{titoh@c.titech.ac.jp}}}
\affil[1]{Department of Mathematical and Computing Science, Tokyo Institute of Technology, Tokyo, Japan}
\date{}
\maketitle

\begin{abstract}
Ripple Effect is a logic puzzle where the player has to fill numbers into empty cells in a rectangular grid. The grid is divided into rooms, and each room must contain consecutive integers starting from 1 to its size. Also, if two cells in the same row or column contain the same number $x$, there must be a space of at least $x$ cells separating the two cells. In this paper, we develop a physical zero-knowledge proof for the Ripple Effect puzzle using a deck of cards, which allows a prover to convince a verifier that he/she knows a solution without revealing it. In particular, given a secret number $x$ and a list of numbers, our protocol can physically verify that $x$ does not appear among the first $x$ numbers in the list without revealing $x$ or any number in the list.

\end{abstract}

\section{Introduction}
\textit{Ripple Effect} (also known as \textit{Hakyu}, \textit{Hakyukoka}, or \textit{Seismic}) is a logic puzzle created by a Japanese company Nikoli, the developer of many famous logic puzzles including Sudoku, Numberlink, and Bridges. A Ripple Effect puzzle consists of a rectangular grid, which is divided into polyominoes called \textit{rooms}. Some cells in the grid comes with a number already in them; we call these cells \textit{fixed cells} and other cells \textit{empty cells}. The objective of this puzzle is to fill a number into each empty cell according to the two following conditions \cite{nikoli}.
\begin{enumerate}
	\item \textit{Room condition}: Each room must contain consecutive integers starting from 1 to its \textit{size} (the number of cells in the room).
	\item \textit{Distance condition}: If two cells in the same row or column contain the same number $x$, there must be a space of at least $x$ cells separating the two cells. See Figure \ref{fig1}.
\end{enumerate}

\begin{figure}
\centering
\begin{tikzpicture}
\draw[step=0.6cm,color={rgb:black,1;white,4}] (0,0) grid (4.2,4.2);

\draw[line width=0.5mm] (0,0) -- (0,4.2);
\draw[line width=0.5mm] (0.6,0.6) -- (0.6,2.4);
\draw[line width=0.5mm] (0.6,3.6) -- (0.6,4.2);
\draw[line width=0.5mm] (1.2,0) -- (1.2,0.6);
\draw[line width=0.5mm] (1.2,1.2) -- (1.2,1.8);
\draw[line width=0.5mm] (1.2,2.4) -- (1.2,4.2);
\draw[line width=0.5mm] (1.8,0.6) -- (1.8,1.8);
\draw[line width=0.5mm] (1.8,2.4) -- (1.8,3.6);
\draw[line width=0.5mm] (2.4,0) -- (2.4,0.6);
\draw[line width=0.5mm] (2.4,1.2) -- (2.4,2.4);
\draw[line width=0.5mm] (2.4,3) -- (2.4,4.2);
\draw[line width=0.5mm] (3,0.6) -- (3,1.2);
\draw[line width=0.5mm] (3,2.4) -- (3,3.6);
\draw[line width=0.5mm] (3.6,0) -- (3.6,0.6);
\draw[line width=0.5mm] (3.6,1.2) -- (3.6,2.4);
\draw[line width=0.5mm] (4.2,0) -- (4.2,4.2);
\draw[line width=0.5mm] (0,0) -- (4.2,0);
\draw[line width=0.5mm] (0.6,0.6) -- (1.2,0.6);
\draw[line width=0.5mm] (1.8,0.6) -- (2.4,0.6);
\draw[line width=0.5mm] (3,0.6) -- (3.6,0.6);
\draw[line width=0.5mm] (1.2,1.2) -- (1.8,1.2);
\draw[line width=0.5mm] (2.4,1.2) -- (3,1.2);
\draw[line width=0.5mm] (3.6,1.2) -- (4.2,1.2);
\draw[line width=0.5mm] (0.6,1.8) -- (3.6,1.8);
\draw[line width=0.5mm] (0,2.4) -- (1.2,2.4);
\draw[line width=0.5mm] (1.8,2.4) -- (2.4,2.4);
\draw[line width=0.5mm] (3,2.4) -- (3.6,2.4);
\draw[line width=0.5mm] (1.2,3) -- (1.8,3);
\draw[line width=0.5mm] (2.4,3) -- (3,3);
\draw[line width=0.5mm] (0.6,3.6) -- (1.2,3.6);
\draw[line width=0.5mm] (1.8,3.6) -- (2.4,3.6);
\draw[line width=0.5mm] (3,3.6) -- (4.2,3.6);
\draw[line width=0.5mm] (0,4.2) -- (4.2,4.2);

\node at (2.1,0.3) {4};
\node at (3.9,0.3) {1};
\node at (2.1,1.5) {5};
\node at (3.9,1.5) {6};
\node at (0.3,2.7) {3};
\node at (1.5,3.3) {2};
\end{tikzpicture}
\hspace{1.2cm}
\begin{tikzpicture}
\draw[step=0.6cm,color={rgb:black,1;white,4}] (0,0) grid (4.2,4.2);

\draw[line width=0.5mm] (0,0) -- (0,4.2);
\draw[line width=0.5mm] (0.6,0.6) -- (0.6,2.4);
\draw[line width=0.5mm] (0.6,3.6) -- (0.6,4.2);
\draw[line width=0.5mm] (1.2,0) -- (1.2,0.6);
\draw[line width=0.5mm] (1.2,1.2) -- (1.2,1.8);
\draw[line width=0.5mm] (1.2,2.4) -- (1.2,4.2);
\draw[line width=0.5mm] (1.8,0.6) -- (1.8,1.8);
\draw[line width=0.5mm] (1.8,2.4) -- (1.8,3.6);
\draw[line width=0.5mm] (2.4,0) -- (2.4,0.6);
\draw[line width=0.5mm] (2.4,1.2) -- (2.4,2.4);
\draw[line width=0.5mm] (2.4,3) -- (2.4,4.2);
\draw[line width=0.5mm] (3,0.6) -- (3,1.2);
\draw[line width=0.5mm] (3,2.4) -- (3,3.6);
\draw[line width=0.5mm] (3.6,0) -- (3.6,0.6);
\draw[line width=0.5mm] (3.6,1.2) -- (3.6,2.4);
\draw[line width=0.5mm] (4.2,0) -- (4.2,4.2);
\draw[line width=0.5mm] (0,0) -- (4.2,0);
\draw[line width=0.5mm] (0.6,0.6) -- (1.2,0.6);
\draw[line width=0.5mm] (1.8,0.6) -- (2.4,0.6);
\draw[line width=0.5mm] (3,0.6) -- (3.6,0.6);
\draw[line width=0.5mm] (1.2,1.2) -- (1.8,1.2);
\draw[line width=0.5mm] (2.4,1.2) -- (3,1.2);
\draw[line width=0.5mm] (3.6,1.2) -- (4.2,1.2);
\draw[line width=0.5mm] (0.6,1.8) -- (3.6,1.8);
\draw[line width=0.5mm] (0,2.4) -- (1.2,2.4);
\draw[line width=0.5mm] (1.8,2.4) -- (2.4,2.4);
\draw[line width=0.5mm] (3,2.4) -- (3.6,2.4);
\draw[line width=0.5mm] (1.2,3) -- (1.8,3);
\draw[line width=0.5mm] (2.4,3) -- (3,3);
\draw[line width=0.5mm] (0.6,3.6) -- (1.2,3.6);
\draw[line width=0.5mm] (1.8,3.6) -- (2.4,3.6);
\draw[line width=0.5mm] (3,3.6) -- (4.2,3.6);
\draw[line width=0.5mm] (0,4.2) -- (4.2,4.2);

\node at (0.3,0.3) {3};
\node at (0.9,0.3) {2};
\node at (1.5,0.3) {5};
\node at (2.1,0.3) {4};
\node at (2.7,0.3) {3};
\node at (3.3,0.3) {2};
\node at (3.9,0.3) {1};
\node at (0.3,0.9) {5};
\node at (0.9,0.9) {1};
\node at (1.5,0.9) {2};
\node at (2.1,0.9) {1};
\node at (2.7,0.9) {4};
\node at (3.3,0.9) {3};
\node at (3.9,0.9) {5};
\node at (0.3,1.5) {4};
\node at (0.9,1.5) {3};
\node at (1.5,1.5) {1};
\node at (2.1,1.5) {5};
\node at (2.7,1.5) {2};
\node at (3.3,1.5) {4};
\node at (3.9,1.5) {6};
\node at (0.3,2.1) {1};
\node at (0.9,2.1) {2};
\node at (1.5,2.1) {4};
\node at (2.1,2.1) {3};
\node at (2.7,2.1) {5};
\node at (3.3,2.1) {1};
\node at (3.9,2.1) {2};
\node at (0.3,2.7) {3};
\node at (0.9,2.7) {4};
\node at (1.5,2.7) {1};
\node at (2.1,2.7) {2};
\node at (2.7,2.7) {3};
\node at (3.3,2.7) {5};
\node at (3.9,2.7) {4};
\node at (0.3,3.3) {1};
\node at (0.9,3.3) {5};
\node at (1.5,3.3) {2};
\node at (2.1,3.3) {4};
\node at (2.7,3.3) {1};
\node at (3.3,3.3) {3};
\node at (3.9,3.3) {1};
\node at (0.3,3.9) {2};
\node at (0.9,3.9) {1};
\node at (1.5,3.9) {3};
\node at (2.1,3.9) {1};
\node at (2.7,3.9) {4};
\node at (3.3,3.9) {2};
\node at (3.9,3.9) {3};
\end{tikzpicture}
\caption{An example of a Ripple Effect puzzle (left) and its solution (right)}
\label{fig1}
\end{figure}
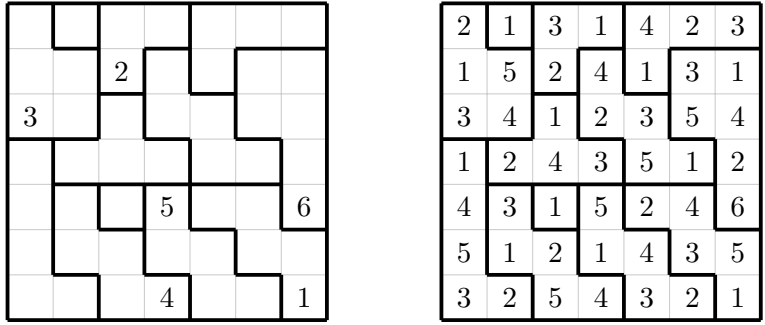

Patricia, an expert in Ripple Effect, created a difficult Ripple Effect puzzle and challenged her friend Victor to solve it. After a while, Victor could not solve her puzzle and started to doubt whether the puzzle has a solution. Patricia needs to convince him that her puzzle indeed has a solution without showing it (which would render the whole challenge pointless). In this situation, Patricia needs a \textit{zero-knowledge proof (ZKP)} to convince Victor.

\subsection{Zero-Knowledge Proof}
A ZKP is an interactive protocol between a prover $P$ and a verifier $V$. Both $P$ and $V$ are given a computational problem $x$, but only $P$ knows a solution $w$ of $x$. A ZKP enables $P$ to convince $V$ that he/she knows $w$ without revealing any information about it. A ZKP must satisfy the following three properties.
\begin{enumerate}
	\item \textbf{Completeness:} If $P$ knows $w$, then $V$ accepts with high probability. (Here we consider the \textit{perfect completeness} property where the probability of acceptance is one.)
	\item \textbf{Soundness:} If $P$ does not know $w$, then $V$ accepts with low probability. (Here we consider the \textit{perfect soundness} property where the probability of acceptance is zero.)
	\item \textbf{Zero-knowledge:} During the verification, $V$ gets no extra information about $w$. Formally, there exists a probabilistic polynomial time algorithm $S$ (called a \textit{simulator}), without an access to $P$ but with a black-box access to $V$, such that the outputs of $S$ follow the same probability distribution as the outputs of the actual protocol.
\end{enumerate}

Goldwasser et al. \cite{zkp0} was the first one to introduce the concept of a ZKP, and Goldreich et al. \cite{zkp} later proved that a computational ZKP exists for every NP problem. As Ripple Effect has been proved to be NP-complete \cite{np}, one can construct a computational ZKP for it. However, such construction is not intuitive or practical as it requires cryptographic primitives.

Instead, many previous results focused on constructing physical ZKPs using a deck of cards. These protocols have benefits that they use only portable objects without requiring computers and also allow external observers to check that the prover truthfully executes the protocol (which is often a challenging task for digital protocols). They also have a great didactic value.

\subsection{Related Work}
Development of card-based ZKPs for logic puzzles began in 2009 with a ZKP for Sudoku of Gradwohl et al. \cite{sudoku0}. However, each of several variants of their protocol either requires special tools or has a nonzero soundness error. Later, Sasaki et al. \cite{sudoku} improved the ZKP for Sudoku to achieve perfect soundness property without using special tools. Recently, Ruangwises \cite{sudoku2} also developed another ZKP for Sudoku using a deck of all different cards. Besides Sudoku, ZKPs for other (mostly Nikoli) logic puzzles have also been developed, including Akari \cite{akari}, Bridges \cite{bridges}, Hitori \cite{nurikabe}, Juosan \cite{takuzu}, Kakuro \cite{akari,kakuro}, KenKen \cite{akari}, Makaro \cite{makaro}, Nonogram \cite{nonogram}, Norinori \cite{norinori}, Numberlink \cite{numberlink}, Nurikabe \cite{nurikabe}, Slitherlink \cite{slitherlink}, Suguru \cite{suguru}, and Takuzu \cite{akari,takuzu}.

The underlying techniques employed by these protocols can physically verify specific number-related functions. For example, a subprotocol in \cite{makaro} verifies that two given numbers are different without revealing their values, and a subprotocol in \cite{sudoku0} verifies that a given list is a permutation of all given numbers in some order without revealing their order.

\subsection{Our Contribution}
In this paper, we develop a physical ZKP with perfect completeness and perfect soundness properties for Ripple Effect puzzle using a deck of cards. More importantly, we extend the set of functions that are known to be physically verifiable. In particular, we develop a physical protocol that, given a secret number $x$ and a list of numbers, verifies that $x$ does not appear among the first $x$ numbers in the list without revealing $x$ or any number in the list.

Unlike the functions verified by many protocols in previous work (see the full analysis in Section \ref{analysis}), the function our protocol has to verify uses cardinal numbers as private information; it uses the secret value $x$ to determine how many elements in the list the condition is imposed on. Therefore, this function is significantly harder to verify \textit{without revealing $x$}, and thus requires novel techniques not used before in other protocols. We consider this result to be an important step in card-based cryptography.

\section{Preliminaries}
\subsection{Cards}
Each card used in our protocol has either $\clubsuit$ or $\heartsuit$ on the front side. All cards have indistinguishable back sides.

For $1 \leq x \leq y$, define $E_y(x)$ to be a sequence of consecutive $y$ cards arranged horizontally, with all of them being \mybox{$\clubsuit$} except the $x$-th card from the left being \mybox{$\heartsuit$}, e.g. $E_3(3)$ is \mbox{\mybox{$\clubsuit$}\mybox{$\clubsuit$}\mybox{$\heartsuit$}} and $E_4(2)$ is \mbox{\mybox{$\clubsuit$}\mybox{$\heartsuit$}\mybox{$\clubsuit$}\mybox{$\clubsuit$}}. We use $E_y(x)$ to encode a number $x$ in a situation where the maximum possible number is at most $y$. This encoding rule was first considered by Shinagawa et al. \cite{polygon} in the context of using a regular $y$-gon card to encode each integer in $\mathbb{Z}/y\mathbb{Z}$. Also, we define $E_y(0)$ to be a sequence of consecutive $y$ cards, all of them being \mybox{$\clubsuit$}. We use $E_y(0)$ to encode a number 0.

Normally, the cards in $E_y(x)$ are arranged horizontally as defined above unless stated otherwise. In some situations, however, we may arrange the cards vertically, where the leftmost card will become the topmost card and the rightmost card will become the bottommost card.

\subsection{Matrix}
We construct an $a \times b$ \textit{matrix} of cards. Let Row $i$ denote the $i$-th topmost row, and Column $j$ denote the $j$-th leftmost column. Let $M(i,j)$ denote the card at Row $i$ and Column $j$ of a matrix $M$. See Figure \ref{fig2}.

\begin{figure}[H]
\centering
\begin{tikzpicture}
\node at (0,0) {\mybox{?}};
\node at (0.5,0) {\mybox{?}};
\node at (1,0) {\mybox{?}};
\node at (1.5,0) {\mybox{?}};
\node at (2,0) {\mybox{?}};
\node at (2.5,0) {\mybox{?}};

\node at (0,0.6) {\mybox{?}};
\node at (0.5,0.6) {\mybox{?}};
\node at (1,0.6) {\mybox{?}};
\node at (1.5,0.6) {\mybox{?}};
\node at (2,0.6) {\mybox{?}};
\node at (2.5,0.6) {\mybox{?}};

\node at (0,1.2) {\mybox{?}};
\node at (0.5,1.2) {\mybox{?}};
\node at (1,1.2) {\mybox{?}};
\node at (1.5,1.2) {\mybox{?}};
\node at (2,1.2) {\mybox{?}};
\node at (2.5,1.2) {\mybox{?}};

\node at (0,1.8) {\mybox{?}};
\node at (0.5,1.8) {\mybox{?}};
\node at (1,1.8) {\mybox{?}};
\node at (1.5,1.8) {\mybox{?}};
\node at (2,1.8) {\mybox{?}};
\node at (2.5,1.8) {\mybox{?}};

\node at (0,2.4) {\mybox{?}};
\node at (0.5,2.4) {\mybox{?}};
\node at (1,2.4) {\mybox{?}};
\node at (1.5,2.4) {\mybox{?}};
\node at (2,2.4) {\mybox{?}};
\node at (2.5,2.4) {\mybox{?}};

\draw[] (-0.3,-0.3) -- (-0.3,3.9);
\draw[] (-1.8,2.8) -- (2.8,2.8);

\node at (-0.6,0) {5};
\node at (-0.6,0.6) {4};
\node at (-0.6,1.2) {3};
\node at (-0.6,1.8) {2};
\node at (-0.6,2.4) {1};
\node at (-1.3,1.2) {Row};

\node at (0,3.1) {1};
\node at (0.5,3.1) {2};
\node at (1,3.1) {3};
\node at (1.5,3.1) {4};
\node at (2,3.1) {5};
\node at (2.5,3.1) {6};
\node at (1.25,3.6) {Column};
\end{tikzpicture}
\caption{An example of a $5 \times 6$ matrix}
\label{fig2}
\end{figure}

\subsection{Pile-Shifting Shuffle}
In a \textit{pile-shifting shuffle}, we rearrange the columns of the matrix by a random cyclic permutation, i.e. move every Column $\ell$ to Column $\ell+r$ for a uniformly random $r \in \{0,1,...,$ $b-1\}$ (where Column $\ell'$ means Column $\ell'-b$ for $\ell'>b$). See Figure \ref{fig3}.

\begin{figure}[H]
\centering
\begin{tikzpicture}
\node at (0,0) {\mybox{?}};
\node at (0.5,0) {\mybox{?}};
\node at (1,0) {\mybox{?}};
\node at (1.5,0) {\mybox{?}};
\node at (2,0) {\mybox{?}};
\node at (2.5,0) {\mybox{?}};

\node at (0,0.6) {\mybox{?}};
\node at (0.5,0.6) {\mybox{?}};
\node at (1,0.6) {\mybox{?}};
\node at (1.5,0.6) {\mybox{?}};
\node at (2,0.6) {\mybox{?}};
\node at (2.5,0.6) {\mybox{?}};

\node at (0,1.2) {\mybox{?}};
\node at (0.5,1.2) {\mybox{?}};
\node at (1,1.2) {\mybox{?}};
\node at (1.5,1.2) {\mybox{?}};
\node at (2,1.2) {\mybox{?}};
\node at (2.5,1.2) {\mybox{?}};

\node at (0,1.8) {\mybox{?}};
\node at (0.5,1.8) {\mybox{?}};
\node at (1,1.8) {\mybox{?}};
\node at (1.5,1.8) {\mybox{?}};
\node at (2,1.8) {\mybox{?}};
\node at (2.5,1.8) {\mybox{?}};

\node at (0,2.4) {\mybox{?}};
\node at (0.5,2.4) {\mybox{?}};
\node at (1,2.4) {\mybox{?}};
\node at (1.5,2.4) {\mybox{?}};
\node at (2,2.4) {\mybox{?}};
\node at (2.5,2.4) {\mybox{?}};

\node at (-0.4,0) {5};
\node at (-0.4,0.6) {4};
\node at (-0.4,1.2) {3};
\node at (-0.4,1.8) {2};
\node at (-0.4,2.4) {1};

\node at (0,2.9) {1};
\node at (0.5,2.9) {2};
\node at (1,2.9) {3};
\node at (1.5,2.9) {4};
\node at (2,2.9) {5};
\node at (2.5,2.9) {6};

\node at (3.4,1.2) {\LARGE{$\Rightarrow$}};
\end{tikzpicture}
\begin{tikzpicture}
\node at (0,0) {\mybox{?}};
\node at (0.5,0) {\mybox{?}};
\node at (1,0) {\mybox{?}};
\node at (1.5,0) {\mybox{?}};
\node at (2,0) {\mybox{?}};
\node at (2.5,0) {\mybox{?}};

\node at (0,0.6) {\mybox{?}};
\node at (0.5,0.6) {\mybox{?}};
\node at (1,0.6) {\mybox{?}};
\node at (1.5,0.6) {\mybox{?}};
\node at (2,0.6) {\mybox{?}};
\node at (2.5,0.6) {\mybox{?}};

\node at (0,1.2) {\mybox{?}};
\node at (0.5,1.2) {\mybox{?}};
\node at (1,1.2) {\mybox{?}};
\node at (1.5,1.2) {\mybox{?}};
\node at (2,1.2) {\mybox{?}};
\node at (2.5,1.2) {\mybox{?}};

\node at (0,1.8) {\mybox{?}};
\node at (0.5,1.8) {\mybox{?}};
\node at (1,1.8) {\mybox{?}};
\node at (1.5,1.8) {\mybox{?}};
\node at (2,1.8) {\mybox{?}};
\node at (2.5,1.8) {\mybox{?}};

\node at (0,2.4) {\mybox{?}};
\node at (0.5,2.4) {\mybox{?}};
\node at (1,2.4) {\mybox{?}};
\node at (1.5,2.4) {\mybox{?}};
\node at (2,2.4) {\mybox{?}};
\node at (2.5,2.4) {\mybox{?}};

\node at (-0.4,0) {5};
\node at (-0.4,0.6) {4};
\node at (-0.4,1.2) {3};
\node at (-0.4,1.8) {2};
\node at (-0.4,2.4) {1};

\node at (0,2.9) {5};
\node at (0.5,2.9) {6};
\node at (1,2.9) {1};
\node at (1.5,2.9) {2};
\node at (2,2.9) {3};
\node at (2.5,2.9) {4};
\end{tikzpicture}
\caption{An example of a pile-shifting shuffle on a $5 \times 6$ matrix with $r=2$}
\label{fig3}
\end{figure}

The pile-shifting shuffle was developed by Shinagawa et al. \cite{polygon}. One can perform this protocol in real world by putting the cards in each column into an envelope and then applying a \textit{Hindu cut}, a basic shuffling operation commonly used in card games \cite{hindu}, to the sequence of envelopes.

\subsection{Rearrangement Protocol}
The sole purpose of a \textit{rearrangement protocol} is to revert columns of a matrix back to their original positions (after we perform pile-shifting shuffles) so that we can reuse all cards in the matrix without revealing them. Slightly different variants of this protocol were used in some previous card-based protocols \cite{makaro,revert1,revert2,numberlink,sudoku}.

Note that throughout our main protocol, we always put $E_b(1)$ in Row 1 when constructing a new matrix, hence we want to ensure that a \mybox{$\heartsuit$} in Row 1 moves back to Column 1. We apply the rearrangement protocol on an $a \times b$ matrix by publicly performing the following steps.

\begin{enumerate}
	\item Apply the pile-shifting shuffle to the matrix.
	\item Turn over all cards in Row 1. Locate the position of a \mybox{$\heartsuit$}. Suppose it is at Column $j$. Turn over all face-up cards.
	\item Shift the columns of the matrix to the left by $j-1$ columns, i.e. move every Column $\ell$ to Column $\ell-(j-1)$ (where Column $\ell'$ means Column $\ell'+b$ for $\ell'<1$).
\end{enumerate}

\subsection{Uniqueness Verification Protocol} \label{unique}
Suppose we have sequences $S_0,S_1,...,S_a$, each consisting of $b$ cards. $S_0$ encodes a positive number, while $S_1,S_2,...,S_a$ encode nonnegative numbers. Our objective is to verify that none of the sequences $S_1,S_2,...,S_a$ encodes the same number as $S_0$ without revealing any encoded number. This protocol is a special case of a protocol developed by Ruangwises and Itoh \cite{numberlink} to count the number of indices $i$ such that $S_i$ encodes the same number as $S_0$. We apply the uniqueness verification protocol by publicly performing the following steps.

\begin{enumerate}
	\item Construct an $(a+2) \times b$ matrix with Row 1 consisting of a sequence $E_b(1)$ and each Row $i+2$ ($i=0,1,...,a$) consisting of the sequence $S_i$.
	\item Apply the pile-shifting shuffle to the matrix.
	\item Turn over all cards in Row 2. Locate the position of a \mybox{$\heartsuit$}. Suppose it is at Column $j$.\footnote{Note that this step also requires that $S_0$ is in a correct format ($E_b(x)$ for some positive integer $x \leq b$). If $V$ sees anything other than one \mybox{$\heartsuit$} and $b-1$ \mybox{$\clubsuit$}s, then $V$ immediately rejects in this step.}
	\item Turn over all cards in Column $j$ from Row 3 to Row $a+2$. If there is no \mybox{$\heartsuit$} among them, then the protocol continues. Otherwise, $V$ rejects and the protocol terminates.
	\item Turn over all face-up cards.
\end{enumerate}

\subsection{Pile-Scramble Shuffle}
In a \textit{pile-scramble shuffle}, we rearrange the columns of the matrix by a random permutation, i.e. move every Column $j$ to Column $p_j$ for a uniformly random permutation $p=(p_1,p_2,...,p_b)$ of $(1,2,...,b)$. See Figure \ref{fig4}.

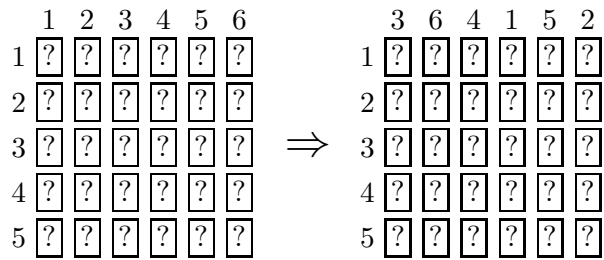
\begin{figure}[H]
\centering
\begin{tikzpicture}
\node at (0,0) {\mybox{?}};
\node at (0.5,0) {\mybox{?}};
\node at (1,0) {\mybox{?}};
\node at (1.5,0) {\mybox{?}};
\node at (2,0) {\mybox{?}};
\node at (2.5,0) {\mybox{?}};

\node at (0,0.6) {\mybox{?}};
\node at (0.5,0.6) {\mybox{?}};
\node at (1,0.6) {\mybox{?}};
\node at (1.5,0.6) {\mybox{?}};
\node at (2,0.6) {\mybox{?}};
\node at (2.5,0.6) {\mybox{?}};

\node at (0,1.2) {\mybox{?}};
\node at (0.5,1.2) {\mybox{?}};
\node at (1,1.2) {\mybox{?}};
\node at (1.5,1.2) {\mybox{?}};
\node at (2,1.2) {\mybox{?}};
\node at (2.5,1.2) {\mybox{?}};

\node at (0,1.8) {\mybox{?}};
\node at (0.5,1.8) {\mybox{?}};
\node at (1,1.8) {\mybox{?}};
\node at (1.5,1.8) {\mybox{?}};
\node at (2,1.8) {\mybox{?}};
\node at (2.5,1.8) {\mybox{?}};

\node at (0,2.4) {\mybox{?}};
\node at (0.5,2.4) {\mybox{?}};
\node at (1,2.4) {\mybox{?}};
\node at (1.5,2.4) {\mybox{?}};
\node at (2,2.4) {\mybox{?}};
\node at (2.5,2.4) {\mybox{?}};

\node at (-0.4,0) {5};
\node at (-0.4,0.6) {4};
\node at (-0.4,1.2) {3};
\node at (-0.4,1.8) {2};
\node at (-0.4,2.4) {1};

\node at (0,2.9) {1};
\node at (0.5,2.9) {2};
\node at (1,2.9) {3};
\node at (1.5,2.9) {4};
\node at (2,2.9) {5};
\node at (2.5,2.9) {6};

\node at (3.4,1.2) {\LARGE{$\Rightarrow$}};
\end{tikzpicture}
\begin{tikzpicture}
\node at (0,0) {\mybox{?}};
\node at (0.5,0) {\mybox{?}};
\node at (1,0) {\mybox{?}};
\node at (1.5,0) {\mybox{?}};
\node at (2,0) {\mybox{?}};
\node at (2.5,0) {\mybox{?}};

\node at (0,0.6) {\mybox{?}};
\node at (0.5,0.6) {\mybox{?}};
\node at (1,0.6) {\mybox{?}};
\node at (1.5,0.6) {\mybox{?}};
\node at (2,0.6) {\mybox{?}};
\node at (2.5,0.6) {\mybox{?}};

\node at (0,1.2) {\mybox{?}};
\node at (0.5,1.2) {\mybox{?}};
\node at (1,1.2) {\mybox{?}};
\node at (1.5,1.2) {\mybox{?}};
\node at (2,1.2) {\mybox{?}};
\node at (2.5,1.2) {\mybox{?}};

\node at (0,1.8) {\mybox{?}};
\node at (0.5,1.8) {\mybox{?}};
\node at (1,1.8) {\mybox{?}};
\node at (1.5,1.8) {\mybox{?}};
\node at (2,1.8) {\mybox{?}};
\node at (2.5,1.8) {\mybox{?}};

\node at (0,2.4) {\mybox{?}};
\node at (0.5,2.4) {\mybox{?}};
\node at (1,2.4) {\mybox{?}};
\node at (1.5,2.4) {\mybox{?}};
\node at (2,2.4) {\mybox{?}};
\node at (2.5,2.4) {\mybox{?}};

\node at (-0.4,0) {5};
\node at (-0.4,0.6) {4};
\node at (-0.4,1.2) {3};
\node at (-0.4,1.8) {2};
\node at (-0.4,2.4) {1};

\node at (0,2.9) {3};
\node at (0.5,2.9) {6};
\node at (1,2.9) {4};
\node at (1.5,2.9) {1};
\node at (2,2.9) {5};
\node at (2.5,2.9) {2};
\end{tikzpicture}
\caption{An example of a pile-scramble shuffle on a $5 \times 6$ matrix}
\label{fig4}
\end{figure}

The pile-scramble shuffle was developed by Ishikawa et al. \cite{scramble}. One can perform this protocol in real world by putting the cards in each column into an envelope and then scrambling all envelopes together completely randomly.

\section{Main Protocol}
Let $m \times n$ be the size of the Ripple Effect grid, and let $k$ be the size of the biggest room. For each fixed cell with a number $x$, the prover $P$ publicly puts a sequence of face-down cards $E_k(x)$ on it. Then, for each empty cell with a number $x$ in $P$'s solution, $P$ secretly puts a sequence of face-down cards $E_k(x)$ on it. $P$ will first verify the distance condition, and then the room condition.

\subsection{Verification Phase for Distance Condition}
The most challenging part of our protocol is to verify the distance condition, which is equivalent to the following statement: for each cell $c$ with a number $x$, the first $x$ cells to the right and to the bottom of $c$ cannot have a number $x$.

First, we will show a protocol to verify that there is no number $x$ among the first $x$ cells to the right of $c$. Then, we apply the same protocol analogously in the direction to the bottom of $c$ as well.

Suppose that cell $c$ has a number $x$. Let $A_0$ be the sequence of cards on $c$. For each $i=1,2,...,k$, let $A_i$ be the sequence of cards on the $i$-th cell to the right of $c$, i.e. $A_1$ is on a cell right next to the right of $c$, $A_2$ is on a second cell to the right of $c$, and so on (if there are only $\ell < k$ cells to the right of $c$, we publicly put $E_k(0)$ in place of $A_i$ for every $i > \ell$).

\subsubsection{Intuition}
The intuition of this protocol is that we will first place the sequences $A_1, A_2,$ $..., A_k$ (each of them arranged vertically) in this order from left to right. Then, we will construct $k-1$ sequences $B_1,B_2,...,B_{k-1}$, all being $E_k(0)$ and arranged vertically, and place them between $A_x$ and $A_{x+1}$ (without revealing $x$). Finally, we will pick $k$ consecutive sequences $A_1,A_2,...,A_x,B_1,B_2,...,B_{k-x}$ and use the uniqueness verification protocol introduced in Section \ref{unique} to verify that none of them encodes the same number as $A_0$. Since $B_1,B_2,...,B_{k-x}$ are all $E_k(0)$, they do not affect the result of the uniqueness verification protocol.

\subsubsection{Formal Steps}
$P$ publicly performs the following steps.

\begin{enumerate}
	\item Construct a $(k+4) \times k$ matrix $M$ by the following procedures. See Figure \ref{fig5}.
	\begin{itemize}
		\item In Row 1, place a sequence $E_k(1)$.
		\item In Row 2, place the sequence $A_0$ (which is $E_k(x)$).
		\item In Row 3, place a sequence $E_k(1)$.
		\item In Row 4, place a sequence $E_k(0)$.
		\item In each Column $j$ ($j=1,2,...,k$), place the sequence $A_j$ arranged vertically from Row 5 to Row $k+4$.
	\end{itemize}
	
\begin{figure}[H]
\centering
\begin{tikzpicture}
\node at (0,-1.1) [rotate=-90]{$A_1$};
\node at (0.5,-1.1) [rotate=-90]{$A_2$};
\node at (1,-1.1) {...};
\node at (1.5,-1.1) [rotate=-90]{$A_k$};

\draw[->] (0,-0.4) -- (0,-0.8);
\draw[->] (0.5,-0.4) -- (0.5,-0.8);
\draw[->] (1.5,-0.4) -- (1.5,-0.8);

\node at (0,0) {\mybox{?}};
\node at (0.5,0) {\mybox{?}};
\node at (1,0) {...};
\node at (1.5,0) {\mybox{?}};

\node at (0,0.7) {\vdots};
\node at (0.5,0.7) {\vdots};
\node at (1,0.7) {\vdots};
\node at (1.5,0.7) {\vdots};

\node at (0,1.2) {\mybox{?}};
\node at (0.5,1.2) {\mybox{?}};
\node at (1,1.2) {...};
\node at (1.5,1.2) {\mybox{?}};

\node at (0,1.8) {\mybox{?}};
\node at (0.5,1.8) {\mybox{?}};
\node at (1,1.8) {...};
\node at (1.5,1.8) {\mybox{?}};

\node at (0,2.7) {\mybox{?}};
\node at (0.5,2.7) {\mybox{?}};
\node at (1,2.7) {...};
\node at (1.5,2.7) {\mybox{?}};
\draw[->] (1.8,2.7) -- (2.2,2.7);
\node at (2.7,2.7) {$E_k(0)$};

\node at (0,3.3) {\mybox{?}};
\node at (0.5,3.3) {\mybox{?}};
\node at (1,3.3) {...};
\node at (1.5,3.3) {\mybox{?}};
\draw[->] (1.8,3.3) -- (2.2,3.3);
\node at (2.7,3.3) {$E_k(1)$};

\node at (0,3.9) {\mybox{?}};
\node at (0.5,3.9) {\mybox{?}};
\node at (1,3.9) {...};
\node at (1.5,3.9) {\mybox{?}};
\draw[->] (1.8,3.9) -- (2.2,3.9);
\node at (2.5,3.9) {$A_0$};

\node at (0,4.5) {\mybox{?}};
\node at (0.5,4.5) {\mybox{?}};
\node at (1,4.5) {...};
\node at (1.5,4.5) {\mybox{?}};
\draw[->] (1.8,4.5) -- (2.2,4.5);
\node at (2.7,4.5) {$E_k(1)$};

\draw[] (-0.3,-0.3) -- (-0.3,5.9);
\draw[] (-2.1,4.9) -- (1.8,4.9);

\node at (-0.9,0) {$k+4$};
\node at (-0.6,0.7) {\vdots};
\node at (-0.6,1.2) {6};
\node at (-0.6,1.8) {5};
\node at (-0.6,2.7) {4};
\node at (-0.6,3.3) {3};
\node at (-0.6,3.9) {2};
\node at (-0.6,4.5) {1};
\node at (-1.6,2.25) {Row};

\node at (0,5.2) {1};
\node at (0.5,5.2) {2};
\node at (1,5.2) {...};
\node at (1.5,5.2) {$k$};
\node at (0.75,5.7) {Column};
\end{tikzpicture}
\caption{A $(k+4) \times k$ matrix $M$ constructed in Step 1}
\label{fig5}
\end{figure}
	
	\item Apply the pile-shifting shuffle to $M$.
	\item Turn over all cards in Row 2 of $M$. Locate the position of a \mybox{$\heartsuit$}. Suppose it is at Column $j_1$. Turn over all face-up cards.
	\item Shift the columns of $M$ to the right by $k-j_1$ columns, i.e. move every Column $\ell$ to Column $\ell+k-j_1$ (where Column $\ell'$ means Column $\ell'-k$ for $\ell'>k$). Observe that after this step, $A_x$ will locate at the rightmost column.
	\item Divide $M$ into a $2 \times k$ matrix $M_1$ and a $(k+2) \times k$ matrix $M_2$. $M_1$ consists of the topmost two rows of $M$, while $M_2$ consists of everything below $M_1$. Each cell $M(i+2,j)$ ($i,j \geq 1$) of $M$ will become a cell $M_2(i,j)$ of a new matrix $M_2$.
	\item Apply the rearrangement protocol to $M_1$. Observe that we now have $E_k(1)$ in Row 1 and $A_0$ in Row 2 of $M_1$. From now on, we will perform operations only on $M_2$ while $M_1$ will be left unchanged.
	\item Append $k-1$ columns to the right of the matrix $M_2$ by the following procedures, making $M_2$ become a $(k+2) \times (2k-1)$ matrix. See Figure \ref{fig6}.
	\begin{itemize}
		\item In Row 1, place a sequence $E_{k-1}(0)$ from Column $k+1$ to Column $2k-1$.
		\item In Row 2, place a sequence $E_{k-1}(1)$ from Column $k+1$ to Column $2k-1$.
		\item In each Column $k+j$ ($j=1,2,...,k-1$), place a sequence $E_k(0)$ arranged vertically from Row 3 to Row $k+2$. We call this sequence $B_j$.
	\end{itemize}
	
\begin{figure}[H]
\centering
\begin{tikzpicture}
\node at (0,-1.3) [rotate=-90]{$A_{x+1}$};
\node at (0.5,-1.3) [rotate=-90]{$A_{x+2}$};
\node at (1,-1.1) {...};
\node at (1.5,-1.1) [rotate=-90]{$A_x$};
\node at (2.3,-1.1) [rotate=-90]{$B_1$};
\node at (2.8,-1.1) [rotate=-90]{$B_2$};
\node at (3.3,-1.1) {...};
\node at (3.8,-1.3) [rotate=-90]{$B_{k-1}$};

\draw[->] (0,-0.4) -- (0,-0.8);
\draw[->] (0.5,-0.4) -- (0.5,-0.8);
\draw[->] (1.5,-0.4) -- (1.5,-0.8);
\draw[->] (2.3,-0.4) -- (2.3,-0.8);
\draw[->] (2.8,-0.4) -- (2.8,-0.8);
\draw[->] (3.8,-0.4) -- (3.8,-0.8);

\node at (0,0) {\mybox{?}};
\node at (0.5,0) {\mybox{?}};
\node at (1,0) {...};
\node at (1.5,0) {\mybox{?}};
\node at (2.3,0) {\mybox{?}};
\node at (2.8,0) {\mybox{?}};
\node at (3.3,0) {...};
\node at (3.8,0) {\mybox{?}};

\node at (0,0.7) {\vdots};
\node at (0.5,0.7) {\vdots};
\node at (1,0.7) {\vdots};
\node at (1.5,0.7) {\vdots};
\node at (2.3,0.7) {\vdots};
\node at (2.8,0.7) {\vdots};
\node at (3.3,0.7) {\vdots};
\node at (3.8,0.7) {\vdots};

\node at (0,1.2) {\mybox{?}};
\node at (0.5,1.2) {\mybox{?}};
\node at (1,1.2) {...};
\node at (1.5,1.2) {\mybox{?}};
\node at (2.3,1.2) {\mybox{?}};
\node at (2.8,1.2) {\mybox{?}};
\node at (3.3,1.2) {...};
\node at (3.8,1.2) {\mybox{?}};

\node at (0,1.8) {\mybox{?}};
\node at (0.5,1.8) {\mybox{?}};
\node at (1,1.8) {...};
\node at (1.5,1.8) {\mybox{?}};
\node at (2.3,1.8) {\mybox{?}};
\node at (2.8,1.8) {\mybox{?}};
\node at (3.3,1.8) {...};
\node at (3.8,1.8) {\mybox{?}};

\node at (-2.3,2.7) {$E_k(0)$};
\draw[->] (-0.2,2.6) to[out=-150, in=-30] (-1.8,2.6);
\node at (0,2.7) {\mybox{?}};
\node at (0.5,2.7) {\mybox{?}};
\node at (1,2.7) {...};
\node at (1.5,2.7) {\mybox{?}};
\node at (2.3,2.7) {\mybox{?}};
\node at (2.8,2.7) {\mybox{?}};
\node at (3.3,2.7) {...};
\node at (3.8,2.7) {\mybox{?}};
\draw[->] (4.1,2.7) -- (4.5,2.7);
\node at (5.2,2.7) {$E_{k-1}(1)$};

\node at (-2.92,3.3) {$E_k(k-x+1)$};
\draw[->] (-0.2,3.2) to[out=-150, in=-30] (-1.8,3.2);
\node at (0,3.3) {\mybox{?}};
\node at (0.5,3.3) {\mybox{?}};
\node at (1,3.3) {...};
\node at (1.5,3.3) {\mybox{?}};
\node at (2.3,3.3) {\mybox{?}};
\node at (2.8,3.3) {\mybox{?}};
\node at (3.3,3.3) {...};
\node at (3.8,3.3) {\mybox{?}};
\draw[->] (4.1,3.3) -- (4.5,3.3);
\node at (5.2,3.3) {$E_{k-1}(0)$};

\draw[] (-0.3,-0.3) -- (-0.3,5.3);
\draw[] (-2.1,3.7) -- (4.1,3.7);

\node at (-0.9,0) {$k+2$};
\node at (-0.6,0.7) {\vdots};
\node at (-0.6,1.2) {4};
\node at (-0.6,1.8) {3};
\node at (-0.6,2.7) {2};
\node at (-0.6,3.3) {1};
\node at (-1.6,1.65) {Row};

\node at (0,4) {1};
\node at (0.5,4) {2};
\node at (1,4) {...};
\node at (1.5,4) {$k$};
\node at (2.3,4.2) [rotate=-90]{$k+1$};
\node at (2.8,4.2) [rotate=-90]{$k+2$};
\node at (3.3,4) {...};
\node at (3.8,4.3) [rotate=-90]{$2k-1$};
\node at (1.9,5.1) {Column};
\end{tikzpicture}
\caption{$(k+2) \times (2k-1)$ matrix $M_2$ after the modification in Step 7}
\label{fig6}
\end{figure}

	\item Apply the pile-shifting shuffle to $M_2$.
	\item Turn over all cards in Row 1 of $M_2$. Locate the position of a \mybox{$\heartsuit$}. Suppose it is at Column $j_2$. Turn over all face-up cards.
	\item For each $i = 1,2,...,k$, let $S_i$ denote a sequence of card arranged vertically at Column $j_2+i-1$ (where Column $\ell'$ means Column $\ell'-(2k-1)$ for $\ell'>2k-1$) from Row 3 to Row $k+2$ of $M_2$. Observe that $(S_1,S_2,...,S_k) = (A_1,A_2,...,A_x,B_1,B_2,...,B_{k-x})$. Then, construct a $(k+2) \times k$ matrix $N$ with Row 1 consisting of a sequence $E_k(1)$ taken from Row 1 of $M_1$, Row 2 consisting of the sequence $A_0$ taken from Row 2 of $M_1$, and each Row $i+2$ ($i=1,2,...,k$) consisting of the sequence $S_i$ taken from $M_2$.
	\item Apply the uniqueness verification protocol on $N$. The intuition of this step is to verify that none of the sequences $A_1,A_2,...,A_x$ encodes the same number as $A_0$ (while $B_1,B_2,...,B_{k-x}$ are all $E_k(0)$).
	\item Apply the rearrangement protocol on $N$, put $A_0$ back onto $c$, and put $S_1,S_2,...,S_k$ back to their corresponding columns in $M_2$.
	\item Apply the pile-shifting shuffle to $M_2$.
	\item Turn over all cards in Row 2 of $M_2$. Locate the position of a \mybox{$\heartsuit$}. Suppose it is at Column $j_3$. Turn over all face-up cards.
	\item Shift the columns of $M_2$ to the right by $k+1-j_3$ columns, i.e. move every Column $\ell$ to Column $\ell+k+1-j_3$ (where Column $\ell'$ means Column $\ell'-(2k-1)$ for $\ell'>2k-1$). Then, remove Columns $k+1,k+2,...,2k-1$ from $M_2$, making $M_2$ become a $(k+2) \times k$ matrix again. Observe that the columns we just removed are exactly the same $k-1$ columns we previously appended to $M_2$.
	\item Apply the rearrangement protocol on $M_2$ and put the sequences $A_1,A_2,...,A_k$ back onto their corresponding cells on the Ripple Effect grid.
\end{enumerate}

$P$ performs these steps analogously in the direction to the right and bottom of every cell in the grid. If every cell passes the verification, $P$ continues to the verification phase for room condition.

\subsection{Verification Phase for Room Condition}
The room condition of Ripple Effect is exactly the same as that of Makaro, and hence can be verified by a subprotocol in \cite{makaro}. Since this is the final step of our protocol, after we finish verifying each room, we do not have to rearrange cards back to their original positions or put them back onto their cells.

$P$ will verify each room separately. For a room $R$ with size $s$, let $A_1,A_2,...,A_s$ be the sequences of cards on the cells in $R$ in any order. To verify room $R$, $P$ publicly performs the following steps.

\begin{enumerate}
	\item Construct a $k \times s$ matrix $M$ by the following procedures: in each Column $j$ ($j=1,2,...,s$), place the sequence $A_j$ arranged vertically from Row 1 to Row $k$.
	\item Apply the pile-scramble shuffle to $M$.
	\item Turn over all cards in $M$. If all columns of $M$ are a permutation of $E_k(1),$ $E_k(2),...,$ $E_k(s)$ arranged vertically, then the protocol continues. Otherwise, $V$ rejects and the protocol terminates.
\end{enumerate}

$P$ performs these steps for every room. If every room passes the verification, then $V$ accepts.

In total, our protocol uses $kmn+2k^2+4k-2 = \Theta(kmn)$ cards.

\section{Proof of Security}
We will prove the perfect completeness, perfect soundness, and zero-knowledge properties of our protocol. We omit the proofs of the verification phase for room condition as they have been shown in \cite{makaro}.

\begin{lemma}[Perfect Completeness] \label{lem1}
If $P$ knows a solution of the Ripple Effect puzzle, then $V$ always accepts.
\end{lemma}

\begin{proof}
First, we will show that the uniqueness verification protocol will pass if none of $S_1,S_2,...,S_a$ encodes the same number as $S_0$. Suppose that $S_0$ encodes a number $z>0$. A \mybox{$\heartsuit$} in Row 2 will locate at Column $z$. Since none of $S_1,S_2,...,S_a$ encodes the number $z$, all cards below Row 2 in the same column as that \mybox{$\heartsuit$} will be all \mybox{$\clubsuit$}s. This remains true after we rearrange the columns in Step 2. Therefore, the verification in Step 4 will pass.

Now consider the main protocol. Suppose that $P$ knows a solution of the puzzle and places cards on each cell according to his/her solution. The verification phase for room condition will pass \cite{makaro}.

Consider the verification phase for distance condition. In Step 1, a \mybox{$\heartsuit$} in Row 2 is at the same column as $A_x$, and will always be. Therefore, in Step 4, the column containing $A_x$ will move to the rightmost column.

In Step 7, the order of the sequences (which are arranged vertically from Row 3 to Row $k+2$) from the leftmost column to the rightmost column is $A_{x+1},A_{x+2},...,A_k,A_1,A_2,...,A_x,$ $B_1,B_2,...,B_{k-1}$. Also, a \mybox{$\heartsuit$} in Row 1 is at the same column as $A_1$, and a \mybox{$\heartsuit$} in Row 3 is at the same column as $B_1$, and they will always be.

In Step 10, since $A_1$ locates at Column $j$, the sequences $S_1,S_2,...,S_k$ will be exactly $A_1,A_2,...,A_x,B_1,B_2,...,B_{k-x}$ in this order. Therefore, in Step 11, the uniqueness verification protocol will pass because none of the sequences $A_1,A_2,...,A_x$ encodes the number $x$, and $B_1,B_2,...,B_{k-x}$ all encode 0.

Since this is true for every cell and every direction (to the right, left, top, and bottom), the verification phase for distance condition will pass, hence $V$ will always accept.
\end{proof}

\begin{lemma}[Perfect Soundness] \label{lem2}
If $P$ does not know a solution of the Ripple Effect puzzle, then $V$ always rejects.
\end{lemma}

\begin{proof}
First, we will show that the uniqueness verification protocol will fail if at least one of $S_1,S_2,...,S_a$ encodes the same number as $S_0$. Suppose that $S_0$ and $S_d$ ($d>0$) both encode a number $z>0$. A \mybox{$\heartsuit$} in Row 2 will locate at Column $z$. Since $S_d$ also encodes the number $z$, a card in Row $d$ in the same column as that \mybox{$\heartsuit$} will be a \mybox{$\heartsuit$}. This remains true after we rearrange the columns in Step 2. Therefore, the verification in Step 4 will fail.

Now consider the main protocol. Suppose that $P$ does not know a solution of the puzzle.\footnote{Our protocol is a \textit{proof-of-knowledge} because the cards on all cells constitute a solution itself. Formally, we can let a machine, called the \textit{extractor}, that turns over all cards on every cell after $P$ places them. If $P$ can convince $V$, then the extractor can obtain the solution from the cards.} If the cards on some cell are not in a correct format ($E_k(x)$ for some integer $x \leq k$), then the uniqueness verification protocol for that cell will fail in Step 3. So, we assume that the cards on every cell are in a correct format and thus correspond to some number. Since $P$ does not know the solution, these numbers must violate either the room condition or the distance condition. If they violate the room condition, the verification phase for room condition will fail \cite{makaro}.

Suppose that the numbers in the grid violate the distance condition. There must be two cells $c$ and $c'$ in the same row or column having the same number $x$, where $c'$ locates on the right or the bottom of $c$ with $\ell < x$ cells of space between them.

Consider when $P$ performs the verification phase for distance condition for $c$ in the direction towards $c'$. The sequence on the cell $c'$ will be $A_{\ell+1}$ By the same reason as in the proof of Lemma \ref{lem1}, in Step 10, the sequences $S_1,S_2,...,S_k$ will be exactly $A_1,A_2,...,A_x,B_1,B_2,...,$ $B_{k-x}$ in this order and thus include $A_{\ell+1}$. Therefore, in Step 11, the uniqueness verification protocol will fail because $A_\ell$ encodes the number $x$, hence $V$ will always reject.
\end{proof}

\begin{lemma}[Zero-Knowledge] \label{lem3}
During the verification phase, $V$ gets no extra information about $P$'s solution of the Ripple Effect puzzle.
\end{lemma}

\begin{proof}
To prove the zero-knowledge property, it is sufficient to prove that all distributions of the values that appear when $P$ turns over cards can be simulated by a simulator $S$ without knowing $P$'s solution.

\begin{itemize}
	\item In the rearrangement protocol:
	\begin{itemize}
		\item Consider Step 2 where we turn over all cards in Row 1. This occurs right after we applied the pile-shifting shuffle to the matrix. Therefore, a \mybox{$\heartsuit$} has an equal probability to appear at each of the $b$ columns, hence this step can be simulated by $S$ without knowing $P$'s solution.
	\end{itemize}
	
	\item In the uniqueness verification protocol:
	\begin{itemize}
		\item Consider Step 3 where we turn over all cards in Row 2. This occurs right after we applied the pile-shifting shuffle to the matrix. Therefore, a \mybox{$\heartsuit$} has an equal probability to appear at each of the $b$ columns, hence this step can be simulated by $S$ without knowing $P$'s solution.
		\item Consider Step 4 where we turn over all cards in Column $j$ from Row 3 to Row $a+2$. If the verification passes, the cards we turn over must be all \mbox{\mybox{$\clubsuit$}s}, hence this step can be simulated by $S$ without knowing $P$'s solution.
	\end{itemize}
	
	\item In the verification phase for room condition:
	\begin{itemize}
		\item Consider Step 3 where we turn over all cards in Row 2 of $M$. This occurs right after we applied the pile-shifting shuffle to $M$. Therefore, a \mybox{$\heartsuit$} has an equal probability to appear at each of the $k$ columns, hence this step can be simulated by $S$ without knowing $P$'s solution.
		\item Consider Step 9 where we turn over all cards in Row 1 of $M_2$. This occurs right after we applied the pile-shifting shuffle to $M_2$. Therefore, a \mybox{$\heartsuit$} has an equal probability to appear at each of the $2k-1$ columns, hence this step can be simulated by $S$ without knowing $P$'s solution.
		\item Consider Step 14 where we turn over all cards in Row 2 of $M_2$. This occurs right after we applied the pile-shifting shuffle to $M_2$. Therefore, a \mybox{$\heartsuit$} has an equal probability to appear at each of the $2k-1$ columns, hence this step can be simulated by $S$ without knowing $P$'s solution.
	\end{itemize}
\end{itemize}

Therefore, we can conclude that $V$ gets no extra information about $P$'s solution.
\end{proof}

\section{Analysis of Nikoli Puzzles and Related Functions} \label{analysis}
In this section, we classify Nikoli logic puzzles into different types in order to highlight the significance of our protocol. See Table \ref{table1}.

\begin{table}[H]
	\centering
	\begin{tabular}{|c|c|c|}
		\hline
		\textbf{Type of Puzzle} & \textbf{\thead{Type A: Numbers as\\ public information}} & \textbf{\thead{Type B: Numbers as\\ private information}} \\ \hline
		\textbf{Type 0: No number} & \multicolumn{2}{c|}{Norinori, Masyu, LITS, Mid-loop, etc.} \\ \hline
		\textbf{Type 1: Nominal numbers} & Numberlink, Hitori, etc. & Sudoku, Nansuke, etc. \\ \hline
		\textbf{Type 2: Ordinal numbers} &  & Makaro \\ \hline
		\textbf{Type 3: Cardinal numbers} & \thead{Akari, Juosan, Shikaku,\\ Slitherlink, Nurikabe, etc.} & Kakuro, Ripple Effect \\ \hline
	\end{tabular}
	\medskip
	\caption{Examples of Nikoli logic puzzles in each type} \label{table1}
\end{table}

Several Nikoli puzzles, including Norinori, Masyu, LITS, and Mid-loop, do not involve numbers at all. We call these puzzle Type 0. Verifying a solution of a Type 0 puzzle can be done by only verifying functions of nominal numbers, e.g. a subprotocol in a ZKP for Norinori \cite{norinori} verifies that a given number appears in a given list without revealing its position in the list.

Some Nikoli puzzles, such as Sudoku, Numberlink, Hitori, and Nansuke, do involve numbers, but only as nominal data, i.e. numbers are treated as labels or symbols. We call these puzzle Type 1. For puzzles in this type, if we replace every number $x$ by $f(x)$ for any function $f: \mathbb{Z}^+ \rightarrow \mathbb{Z}^+$, the conditions will remain exactly the same. Like Type 0 puzzles, verifying a solution of a Type 1 puzzle can be done by only verifying functions of nominal numbers, e.g. a subprotocol in a ZKP for Sudoku \cite{sudoku0} verifies that a given list is a permutation of all given numbers in some order without revealing their order.

We further divide Type 1 puzzles into two subtypes, 1A and 1B. For Type 1A puzzles, numbers appear only as public information, i.e. already given in the puzzle but not as parts of a solution. Numberlink and Hitori are Type 1A puzzles. For Type 1B puzzles, numbers also appear as private information, i.e. as parts of a solution. Sudoku and Nansuke are Type 1B puzzles. The difference of these two subtypes is that when verifying a solution of a Type 1B puzzle, we cannot reveal any number, which makes developing ZKPs for Type 1B puzzles generally harder.

The next type of Nikoli puzzles is Type 2, which consists of puzzles involving numbers as ordinal data, i.e. numbers are comparable. For puzzles in this type, if we replace every number $x$ by $f(x)$ for any \textit{increasing} function $f: \mathbb{Z}^+ \rightarrow \mathbb{Z}^+$, the conditions will remain exactly the same. We find Makaro as the only Nikoli puzzle classified into this type. In Makaro, numbers are parts of a solution and thus are private information, so we put Makaro into Type 2B. Verifying a solution of a Type 2 puzzle needs to verify functions of ordinal numbers, e.g. a subprotocol in a ZKP for Makaro \cite{makaro} verifies that a given number in a list is the largest one in that list without revealing any number in the list.

Type 3 puzzles are the ones involving cardinal numbers, i.e. using the mathematical meaning of numbers. Most of Nikoli puzzles, including Akari, Juosan, Shikaku, Slitherlink, and Nurikabe, use cardinal numbers only as public information and hence are classified into Type 3A. Verifying a solution of a Type 3A puzzle needs to verify functions of cardinal numbers, e.g. a subprotocol in ZKP for Nurikabe \cite{nurikabe} verifies the existence of $k$ connected white cells in a grid, given a number $k$. Note that during the verification, all numbers are public information and can be revealed, making the verification easier than that of Type 3B puzzles.

Lastly, Type 3B puzzles use cardinal numbers as parts of a solution, and thus as private information. We find Kakuro and Ripple Effect as the only two Nikoli puzzles in this type. However, the special property of Kakuro, which differs from other Type 3 Nikoli puzzles, is that numbers are never actually used in the sense of cardinality (such as counting the numbers of cells) but are only used for addition. Observe that in Kakuro, if we replace every number $x$ by $f(x)$ for any \textit{linear} function $f: \mathbb{Z}^+ \rightarrow \mathbb{Z}^+$, the conditions will remain exactly the same. (This property is not true for all other Type 3 Nikoli puzzles.) In particular, a subprotocol in a ZKP for Kakuro \cite{kakuro} verifies that the sum of all numbers in a list is equal to a given number. This leaves Ripple Effect as the only Nikoli puzzle that actually uses cardinal number as private information.

As a result, we consider verifying a solution of Ripple Effect to be harder than that of other puzzles, and the techniques used in the ZKP for Kakuro cannot be applied to our protocol. In particular, given a secret number $x$ and a list of numbers, our protocol verifies that $x$ does not appear among the first $x$ numbers in the list without revealing $x$ or any number in the list. This function is significantly harder to verify \textit{without revealing $x$}, and thus requires novel techniques not used before in other protocols.

\section{Future Work}
We developed a physical protocol of ZKP for the Ripple Effect puzzle using $\Theta(kmn)$ cards. A possible future work is to develop a protocol for this puzzle that requires asymptotically fewer number of cards, or the one that can be performed using a deck of all different cards.

Another challenging future work is to find other Type 3B (non-Nikoli) logic puzzles and develop a ZKP for such puzzles, which may involve physical verification of more complicated number-related functions.

\end{document}